\documentclass[11pt]{article}
\usepackage{amssymb,amsthm,amsmath,algorithm2e,float,cite,enumerate}
\usepackage{geometry}
\usepackage{tikz}

\newtheorem{theorem}{Theorem}[section]

\newtheorem{claim}{Claim}

\usepackage{setspace}
\usepackage{lineno}

\begin{document}

\onehalfspace

\title{Dynamic monopolies for interval graphs with bounded thresholds}

\author{St\'{e}phane Bessy$^1$ 
\and Stefan Ehard$^2$
\and Lucia D. Penso$^2$
\and Dieter Rautenbach$^2$}

\date{}

\maketitle

\begin{center}
$^1$ Laboratoire d'Informatique, de Robotique et de Micro\'{e}lectronique de Montpellier,\\
Montpellier, France, \texttt{stephane.bessy@lirmm.fr}\\[3mm]
$^2$ Institut f\"{u}r Optimierung und Operations Research, 
Universit\"{a}t Ulm, Ulm, Germany,
\{\texttt{stefan.ehard,lucia.penso,dieter.rautenbach}\}\texttt{@uni-ulm.de}\\[3mm]
\end{center}

\begin{abstract}
For a graph $G$ and an integer-valued threshold function $\tau$ on its
vertex set, a dynamic monopoly is a set of vertices of $G$ such that
iteratively adding to it vertices $u$ of $G$ that have at least
$\tau(u)$ neighbors in it eventually yields the vertex set of $G$.  We
show that the problem of finding a dynamic monopoly of minimum order
can be solved in polynomial time for interval graphs with bounded
threshold functions, but is NP-hard for chordal graphs allowing
unbounded threshold functions.
\end{abstract}

{\small 
\begin{tabular}{lp{13cm}}
\end{tabular}
}


\section{Introduction}

Dynamic monopolies are a simple model for various types of viral
processes in networks \cite{dori,drro,keklta}. Let $G$ be a finite,
simple, and undirected graph.  A {\it threshold function} for $G$ is
an integer-valued function whose domain contains the vertex set $V(G)$
of $G$.  Let $\tau$ be a threshold function for $G$.  For a set $D$ of
vertices of $G$, the {\it hull $H_{(G,\tau)}(D)$ of $D$ in $(G,\tau)$}
is the set obtained by starting with the empty set, and iteratively
adding vertices $u$ to the current set that belong to $D$ or have at
least $\tau(u)$ neighbors in the current set as long as possible.  The
set $D$ is a {\it dynamic monopoly} or a {\it target set} of
$(G,\tau)$ if $H_{(G,\tau)}(D)$ equals $V(G)$, and the minimum order
of a dynamic monopoly of $(G,\tau)$ is denoted by ${\rm dyn}(G,\tau)$.

The parameter ${\rm dyn}(G,\tau)$ is computationally hard even when
restricted to instances with bounded threshold functions
\cite{ch,cedoperasz,drro,kylivy}.  Efficient algorithms that work for
unbounded threshold functions are known for trees
\cite{ch,cedoperasz,drro}, block-cactus graphs \cite{chhuliwuye},
graphs of bounded treewidth \cite{behelone}, and graphs whose blocks
have bounded order \cite{cedoperasz}.  For bounded threshold
functions, some more instances become tractable, and ${\rm
  dyn}(G,\tau)$ can be computed efficiently if $G$ is cubic and
$\tau=2$ \cite{barasasz,kylivy} or if $G$ is chordal and $\tau\leq 2$
\cite{cedoperasz,chhuliwuye}.  The latter result relies on the case
$t=2$ of the following theorem.

\begin{theorem}[Chiang et al. \cite{chhuliwuye}]\label{theoremchordal}
If $t$ is a non-negative integer, $G$ is a $t$-connected chordal
graph, and $\tau$ is a threshold function for $G$ with $\tau(u)\leq t$
for every vertex $u$ of $G$, then ${\rm dyn}(G,\tau)\leq t$.
\end{theorem}

Since this result holds for arbitrary $t$, it suggests that there
might be an efficient algorithm for chordal graphs and bounded
threshold functions.  In the present paper we show that this is at
least true for interval graphs, which form a prominent subclass of
chordal graphs.

\begin{theorem}\label{theorem2}
Let $t$ be a non-negative integer.  For a given interval graph $G$,
and a given threshold function $\tau$ for $G$ with $\tau(u)\leq t$ for
every vertex $u$ of $G$, the value of ${\rm dyn}(G,\tau)$ can be
determined in polynomial time.
\end{theorem}

It is open \cite{chniniwe} whether ${\rm dyn}(G,\tau)$ is fixed
parameter tractable for instances with bounded threshold functions
when parameterized by the distance to interval graphs. Note that
Theorem~\ref{theorem2} would be a consequence of such a fixed
parameter tractability.

As our second result we show that dynamic monopolies remain hard for
chordal graphs with unbounded threshold functions.

\begin{theorem}\label{theorem1}
For a given triple $(G,\tau,k)$, where $G$ is a chordal graph, $\tau$
is a threshold function for $G$, and $k$ is a positive integer, it is
NP-complete to decide whether ${\rm dyn}(G,\tau)\leq k$.
\end{theorem}

\section{Proofs}

Our approach to prove Theorem \ref{theorem2} is to construct a
sequence $G_1\subseteq G_2\subseteq \ldots \subseteq G_k$ of subgraphs
of $G$ in such a way that $G_k=G$, and Theorem \ref{theoremchordal}
implies that every minimum dynamic monopoly $D$ for $(G,\tau)$
intersects a suitable supergraph $\partial G_i$ of each
$G_i-V(G_{i-1})$ in at most $t$ vertices.  This enables us to apply
dynamic programming efficiently calculating partial information for
each $G_i$ by emulating the formation of the hull of $D$ within
$\partial G_i$, and exploit previously computed information for
$G_{i-1}$.  A notion that is useful in this context is the one of a
{\it cascade} for a dynamic monopoly $D$ of $(G,\tau)$, defined as a
linear order $u_1\prec\ldots\prec u_n$ of the vertices of $G$ such
that, for every $i$ in $[n]$, either $u_i\in D$ or $u_i\not\in D$ and
$|N_G(u_i)\cap \{ u_j:j\in [i-1]\}|\geq \tau (u_j)$, where $[k]$
denotes the set of positive integers that are less than or equal to
some integer $k$.  A cascade encodes the order in which the vertices
of $G$ can be added to the hull of $D$ starting with the empty set.
Clearly, every dynamic monopoly admits at least one cascade $\prec$.
Furthermore, we may assume that $u\prec v$ for every $u\in D$ and
every $v\in V(G)\setminus D$.

We proceed to the proof of our first result.

\begin{proof}[Proof of Theorem \ref{theorem2}]
Let $t$, $G$, and $\tau$ be as in the statement.  Clearly, we may
assume that $G$ is connected.  Let $n$ be the order of $G$.  In linear
time \cite{bolu}, we can determine an interval representation $(I(u))_{u\in V(G)}$ of $G$, that is, two distinct vertices $u$
and $v$ of $G$ are adjacent if and only if the intervals $I(u)$ and
$I(v)$ intersect.  By applying well-known manipulations, we may assume
that each interval $I(u)$ is closed, and that the $2n$ endpoints of
the $n$ intervals are all distinct.

Let $x_1<x_2<\ldots<x_{2n}$ be the endpoints of the intervals.  For
each $i\in [2n-1]$, let $C_i$ be the set of vertices $u$ of $G$ with
$I_i:=[x_i,x_{i+1}]\subseteq I(u)$, and let $c_i=|C_i|$.  Since each
$x_i$ is either the right endpoint of exactly one interval or the left
endpoint of exactly one interval, we have $|c_{i+1}-c_i|=1$ for every
$i\in [2n-1]$.

Our first claim states a folklore property of interval graphs; we
include a proof for the sake of completeness.

\begin{claim}\label{claim1}
If $C$ is a minimal vertex cut of $G$, then $C=C_i$ for some $i\in
[2n-2]\setminus \{ 1\}$ with $c_i<\min\{ c_{i-1},c_{i+1}\}$.
\end{claim}
\begin{proof}[Proof of Claim \ref{claim1}]
Clearly, if $i\in [2n-2]\setminus \{ 1\}$ is such that $c_i<\min\{
c_{i-1},c_{i+1}\}$, then $C_i$ is a minimal vertex cut separating the
unique vertex in $C_{i-1}\setminus C_i$ from the unique vertex in
$C_{i+1}\setminus C_i$.  Conversely, let $C$ be a minimal vertex cut
of $G$.  Let $u$ and $v$ be vertices in distinct components of $G-C$.
We may assume that the right endpoint $r(u)$ of $I(u)$ is less than
the left endpoint $\ell(v)$ of $I(v)$.  There are indices $i_1$ and
$i_2$ such that $[r(u),\ell(v)]=\bigcup\limits_{j=i_1}^{i_2}I_j$.
Since $G-C$ contains no path between $u$ and $v$, there is some index
$i$ with $i_1\leq i\leq i_2$ and $C_i\subseteq C$.  Since $G-C_i$
contains no path between $u$ and $v$, the minimality of $C$ implies
$C\subseteq C_i$, and, hence, $C=C_i$.  If $i=i_1$, then
$c_i<c_{i-1}$, because $I(u)$ ends in $i_1$.  If $i>i_1$ and
$c_i>c_{i-1}$, then $C_{i-1}$ is a proper subset of $C_i$, and also
$G-C_{i-1}$ contains no path between $u$ and $v$, contradicting the
minimality of $C$.  Therefore, $c_i<c_{i-1}$, and, by symmetry, also
$c_i<c_{i+1}$.
\end{proof}

Let $j_1<j_2<\ldots<j_{k-1}$ be the indices $i$ in $[2n-1]\setminus \{
1\}$ with $c_i<\min\{ c_{i-1},c_{i+1},t\}$, and let $j_k=2n-1$.  For
$i\in [k]$, let $G_i$ be the subgraph of $G$ induced by
$V_i:=C_1\cup\cdots\cup C_{j_i}$, and let $B_i=C_{j_i}$.  Note that
$B_i$ contains all vertices in $V_i$ that have a neighbor in
$V(G)\setminus V_i$, and that $|B_i|<t$.  Let $\partial V_1=V_1$, and,
for $i\in [k]\setminus \{ 1\}$, let $\partial V_i=(V_i\setminus
V_{i-1})\cup B_{i-1}$.  For $i\in [k]$, let $\partial G_i$ be the
subgraph of $G$ induced by $\partial V_i$, cf. Figure~\ref{fig:int-rep}.

\begin{figure}[t]
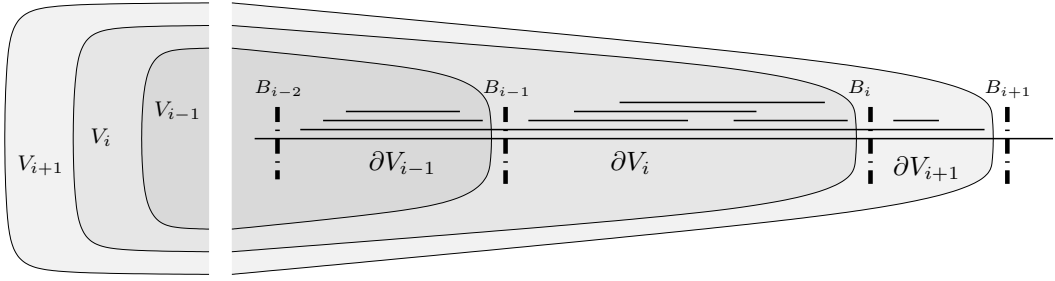

\begin{center}
\ifx\XFigwidth\undefined\dimen1=0pt\else\dimen1\XFigwidth\fi
\divide\dimen1 by 10374
\ifx\XFigheight\undefined\dimen3=0pt\else\dimen3\XFigheight\fi
\divide\dimen3 by 2724
\ifdim\dimen1=0pt\ifdim\dimen3=0pt\dimen1=2486sp\dimen3\dimen1
  \else\dimen1\dimen3\fi\else\ifdim\dimen3=0pt\dimen3\dimen1\fi\fi
\tikzpicture[x=+\dimen1, y=+\dimen3]
{\ifx\XFigu\undefined\catcode`\@11
\def\temp{\alloc@1\dimen\dimendef\insc@unt}\temp\XFigu\catcode`\@12\fi}
\XFigu2486sp
\ifdim\XFigu<0pt\XFigu-\XFigu\fi
\definecolor{xfigc32}{rgb}{0.612,0.000,0.000}
\definecolor{xfigc33}{rgb}{0.549,0.549,0.549}
\definecolor{xfigc34}{rgb}{0.549,0.549,0.549}
\definecolor{xfigc35}{rgb}{0.259,0.259,0.259}
\definecolor{xfigc36}{rgb}{0.549,0.549,0.549}
\definecolor{xfigc37}{rgb}{0.259,0.259,0.259}
\definecolor{xfigc38}{rgb}{0.549,0.549,0.549}
\definecolor{xfigc39}{rgb}{0.259,0.259,0.259}
\definecolor{xfigc40}{rgb}{0.549,0.549,0.549}
\definecolor{xfigc41}{rgb}{0.259,0.259,0.259}
\definecolor{xfigc42}{rgb}{0.549,0.549,0.549}
\definecolor{xfigc43}{rgb}{0.259,0.259,0.259}
\definecolor{xfigc44}{rgb}{0.557,0.557,0.557}
\definecolor{xfigc45}{rgb}{0.761,0.761,0.761}
\definecolor{xfigc46}{rgb}{0.431,0.431,0.431}
\definecolor{xfigc47}{rgb}{0.267,0.267,0.267}
\definecolor{xfigc48}{rgb}{0.557,0.561,0.557}
\definecolor{xfigc49}{rgb}{0.443,0.443,0.443}
\definecolor{xfigc50}{rgb}{0.682,0.682,0.682}
\definecolor{xfigc51}{rgb}{0.200,0.200,0.200}
\definecolor{xfigc52}{rgb}{0.580,0.576,0.584}
\definecolor{xfigc53}{rgb}{0.455,0.439,0.459}
\definecolor{xfigc54}{rgb}{0.333,0.333,0.333}
\definecolor{xfigc55}{rgb}{0.702,0.702,0.702}
\definecolor{xfigc56}{rgb}{0.765,0.765,0.765}
\definecolor{xfigc57}{rgb}{0.427,0.427,0.427}
\definecolor{xfigc58}{rgb}{0.271,0.271,0.271}
\definecolor{xfigc59}{rgb}{0.886,0.886,0.933}
\definecolor{xfigc60}{rgb}{0.580,0.580,0.604}
\definecolor{xfigc61}{rgb}{0.859,0.859,0.859}
\definecolor{xfigc62}{rgb}{0.631,0.631,0.718}
\definecolor{xfigc63}{rgb}{0.929,0.929,0.929}
\definecolor{xfigc64}{rgb}{0.878,0.878,0.878}
\definecolor{xfigc65}{rgb}{0.525,0.675,1.000}
\definecolor{xfigc66}{rgb}{0.439,0.439,1.000}
\definecolor{xfigc67}{rgb}{0.776,0.718,0.592}
\definecolor{xfigc68}{rgb}{0.937,0.973,1.000}
\definecolor{xfigc69}{rgb}{0.863,0.796,0.651}
\definecolor{xfigc70}{rgb}{0.251,0.251,0.251}
\definecolor{xfigc71}{rgb}{0.502,0.502,0.502}
\definecolor{xfigc72}{rgb}{0.753,0.753,0.753}
\definecolor{xfigc73}{rgb}{0.667,0.667,0.667}
\definecolor{xfigc74}{rgb}{0.780,0.765,0.780}
\definecolor{xfigc75}{rgb}{0.337,0.318,0.318}
\definecolor{xfigc76}{rgb}{0.843,0.843,0.843}
\definecolor{xfigc77}{rgb}{0.522,0.502,0.490}
\definecolor{xfigc78}{rgb}{0.824,0.824,0.824}
\definecolor{xfigc79}{rgb}{0.227,0.227,0.227}
\definecolor{xfigc80}{rgb}{0.271,0.451,0.667}
\definecolor{xfigc81}{rgb}{0.482,0.475,0.647}
\definecolor{xfigc82}{rgb}{0.451,0.459,0.549}
\definecolor{xfigc83}{rgb}{0.969,0.969,0.969}
\definecolor{xfigc84}{rgb}{0.255,0.271,0.255}
\definecolor{xfigc85}{rgb}{0.388,0.365,0.808}
\definecolor{xfigc86}{rgb}{0.745,0.745,0.745}
\definecolor{xfigc87}{rgb}{0.318,0.318,0.318}
\definecolor{xfigc88}{rgb}{0.906,0.890,0.906}
\definecolor{xfigc89}{rgb}{0.000,0.000,0.286}
\definecolor{xfigc90}{rgb}{0.475,0.475,0.475}
\definecolor{xfigc91}{rgb}{0.188,0.204,0.188}
\definecolor{xfigc92}{rgb}{0.255,0.255,0.255}
\definecolor{xfigc93}{rgb}{0.780,0.714,0.588}
\definecolor{xfigc94}{rgb}{0.867,0.616,0.576}
\definecolor{xfigc95}{rgb}{0.945,0.925,0.878}
\definecolor{xfigc96}{rgb}{0.886,0.784,0.659}
\definecolor{xfigc97}{rgb}{0.882,0.882,0.882}
\definecolor{xfigc98}{rgb}{0.855,0.478,0.102}
\definecolor{xfigc99}{rgb}{0.945,0.894,0.102}
\definecolor{xfigc100}{rgb}{0.533,0.490,0.761}
\definecolor{xfigc101}{rgb}{0.690,0.631,0.576}
\definecolor{xfigc102}{rgb}{0.514,0.486,0.867}
\definecolor{xfigc103}{rgb}{0.839,0.839,0.839}
\definecolor{xfigc104}{rgb}{0.549,0.549,0.647}
\definecolor{xfigc105}{rgb}{0.290,0.290,0.290}
\definecolor{xfigc106}{rgb}{0.549,0.420,0.420}
\definecolor{xfigc107}{rgb}{0.353,0.353,0.353}
\definecolor{xfigc108}{rgb}{0.388,0.388,0.388}
\definecolor{xfigc109}{rgb}{0.718,0.608,0.451}
\definecolor{xfigc110}{rgb}{0.255,0.576,1.000}
\definecolor{xfigc111}{rgb}{0.749,0.439,0.231}
\definecolor{xfigc112}{rgb}{0.859,0.467,0.000}
\definecolor{xfigc113}{rgb}{0.855,0.722,0.000}
\definecolor{xfigc114}{rgb}{0.000,0.392,0.000}
\definecolor{xfigc115}{rgb}{0.353,0.420,0.231}
\definecolor{xfigc116}{rgb}{0.827,0.827,0.827}
\definecolor{xfigc117}{rgb}{0.557,0.557,0.643}
\definecolor{xfigc118}{rgb}{0.953,0.725,0.365}
\definecolor{xfigc119}{rgb}{0.537,0.600,0.420}
\definecolor{xfigc120}{rgb}{0.392,0.392,0.392}
\definecolor{xfigc121}{rgb}{0.718,0.902,1.000}
\definecolor{xfigc122}{rgb}{0.525,0.753,0.925}
\definecolor{xfigc123}{rgb}{0.741,0.741,0.741}
\definecolor{xfigc124}{rgb}{0.827,0.584,0.322}
\definecolor{xfigc125}{rgb}{0.596,0.824,0.996}
\definecolor{xfigc126}{rgb}{0.380,0.380,0.380}
\definecolor{xfigc127}{rgb}{0.682,0.698,0.682}
\definecolor{xfigc128}{rgb}{1.000,0.604,0.000}
\definecolor{xfigc129}{rgb}{0.549,0.612,0.420}
\definecolor{xfigc130}{rgb}{0.969,0.420,0.000}
\definecolor{xfigc131}{rgb}{0.353,0.420,0.224}
\definecolor{xfigc132}{rgb}{0.549,0.612,0.420}
\definecolor{xfigc133}{rgb}{0.549,0.612,0.482}
\definecolor{xfigc134}{rgb}{0.094,0.290,0.094}
\definecolor{xfigc135}{rgb}{0.678,0.678,0.678}
\definecolor{xfigc136}{rgb}{0.969,0.741,0.353}
\definecolor{xfigc137}{rgb}{0.388,0.420,0.612}
\definecolor{xfigc138}{rgb}{0.871,0.000,0.000}
\definecolor{xfigc139}{rgb}{0.678,0.678,0.678}
\definecolor{xfigc140}{rgb}{0.969,0.741,0.353}
\definecolor{xfigc141}{rgb}{0.678,0.678,0.678}
\definecolor{xfigc142}{rgb}{0.969,0.741,0.353}
\definecolor{xfigc143}{rgb}{0.388,0.420,0.612}
\definecolor{xfigc144}{rgb}{0.322,0.420,0.161}
\definecolor{xfigc145}{rgb}{0.580,0.580,0.580}
\definecolor{xfigc146}{rgb}{0.000,0.388,0.000}
\definecolor{xfigc147}{rgb}{0.000,0.388,0.290}
\definecolor{xfigc148}{rgb}{0.482,0.518,0.290}
\definecolor{xfigc149}{rgb}{0.906,0.741,0.482}
\definecolor{xfigc150}{rgb}{0.647,0.710,0.776}
\definecolor{xfigc151}{rgb}{0.420,0.420,0.580}
\definecolor{xfigc152}{rgb}{0.518,0.420,0.420}
\definecolor{xfigc153}{rgb}{0.322,0.612,0.290}
\definecolor{xfigc154}{rgb}{0.839,0.906,0.906}
\definecolor{xfigc155}{rgb}{0.322,0.388,0.388}
\definecolor{xfigc156}{rgb}{0.094,0.420,0.290}
\definecolor{xfigc157}{rgb}{0.612,0.647,0.710}
\definecolor{xfigc158}{rgb}{1.000,0.580,0.000}
\definecolor{xfigc159}{rgb}{1.000,0.580,0.000}
\definecolor{xfigc160}{rgb}{0.000,0.388,0.290}
\definecolor{xfigc161}{rgb}{0.482,0.518,0.290}
\definecolor{xfigc162}{rgb}{0.388,0.451,0.482}
\definecolor{xfigc163}{rgb}{0.906,0.741,0.482}
\definecolor{xfigc164}{rgb}{0.094,0.290,0.094}
\definecolor{xfigc165}{rgb}{0.969,0.741,0.353}
\definecolor{xfigc166}{rgb}{0.000,0.000,0.000}
\definecolor{xfigc167}{rgb}{0.969,0.220,0.161}
\definecolor{xfigc168}{rgb}{0.000,0.000,0.000}
\definecolor{xfigc169}{rgb}{1.000,1.000,0.322}
\definecolor{xfigc170}{rgb}{0.322,0.475,0.290}
\definecolor{xfigc171}{rgb}{0.388,0.604,0.353}
\definecolor{xfigc172}{rgb}{0.776,0.380,0.259}
\definecolor{xfigc173}{rgb}{0.906,0.412,0.259}
\definecolor{xfigc174}{rgb}{1.000,0.475,0.322}
\definecolor{xfigc175}{rgb}{0.871,0.871,0.871}
\definecolor{xfigc176}{rgb}{0.953,0.933,0.827}
\definecolor{xfigc177}{rgb}{0.961,0.682,0.365}
\definecolor{xfigc178}{rgb}{0.584,0.808,0.600}
\definecolor{xfigc179}{rgb}{0.710,0.082,0.490}
\definecolor{xfigc180}{rgb}{0.933,0.933,0.933}
\definecolor{xfigc181}{rgb}{0.518,0.518,0.518}
\definecolor{xfigc182}{rgb}{0.482,0.482,0.482}
\definecolor{xfigc183}{rgb}{0.000,0.353,0.000}
\definecolor{xfigc184}{rgb}{0.906,0.451,0.451}
\definecolor{xfigc185}{rgb}{1.000,0.796,0.192}
\definecolor{xfigc186}{rgb}{0.161,0.475,0.290}
\definecolor{xfigc187}{rgb}{0.871,0.157,0.129}
\definecolor{xfigc188}{rgb}{0.129,0.349,0.776}
\definecolor{xfigc189}{rgb}{0.973,0.973,0.973}
\definecolor{xfigc190}{rgb}{0.902,0.902,0.902}
\definecolor{xfigc191}{rgb}{0.129,0.518,0.353}
\definecolor{xfigc192}{rgb}{0.906,0.906,0.906}
\definecolor{xfigc193}{rgb}{0.443,0.459,0.443}
\definecolor{xfigc194}{rgb}{0.851,0.851,0.851}
\definecolor{xfigc195}{rgb}{0.337,0.620,0.690}
\definecolor{xfigc196}{rgb}{0.788,0.788,0.788}
\definecolor{xfigc197}{rgb}{0.875,0.847,0.875}
\definecolor{xfigc198}{rgb}{0.969,0.953,0.969}
\definecolor{xfigc199}{rgb}{0.800,0.800,0.800}
\clip(-677,-4287) rectangle (9697,-1563);
\tikzset{inner sep=+0pt, outer sep=+0pt}
\pgfsetlinewidth{+7.5\XFigu}
\pgfsetfillcolor{.!5}
\filldraw (1575,-1575)--(1576,-1575)--(1578,-1575)--(1583,-1576)--(1590,-1576)--(1600,-1577)
  --(1615,-1579)--(1633,-1580)--(1656,-1583)--(1685,-1585)--(1719,-1588)--(1758,-1592)
  --(1804,-1596)--(1855,-1601)--(1913,-1607)--(1977,-1613)--(2048,-1619)--(2124,-1626)
  --(2207,-1634)--(2295,-1642)--(2389,-1651)--(2489,-1661)--(2593,-1670)--(2702,-1681)
  --(2815,-1691)--(2932,-1703)--(3052,-1714)--(3175,-1726)--(3301,-1738)--(3428,-1750)
  --(3558,-1762)--(3688,-1775)--(3819,-1787)--(3951,-1800)--(4082,-1812)--(4213,-1825)
  --(4343,-1838)--(4472,-1850)--(4600,-1863)--(4727,-1875)--(4851,-1887)--(4974,-1899)
  --(5094,-1911)--(5212,-1923)--(5328,-1934)--(5441,-1946)--(5551,-1957)--(5659,-1968)
  --(5764,-1978)--(5866,-1989)--(5966,-1999)--(6063,-2009)--(6157,-2019)--(6248,-2028)
  --(6337,-2038)--(6423,-2047)--(6507,-2056)--(6588,-2065)--(6666,-2073)--(6742,-2082)
  --(6816,-2090)--(6887,-2098)--(6956,-2106)--(7023,-2114)--(7088,-2121)--(7151,-2129)
  --(7212,-2136)--(7271,-2144)--(7329,-2151)--(7384,-2158)--(7438,-2165)--(7491,-2172)
  --(7542,-2179)--(7592,-2186)--(7640,-2192)--(7687,-2199)--(7733,-2206)--(7778,-2213)
  --(7866,-2226)--(7950,-2239)--(8031,-2253)--(8107,-2266)--(8180,-2279)--(8250,-2293)
  --(8316,-2306)--(8379,-2320)--(8439,-2333)--(8496,-2347)--(8550,-2360)--(8601,-2374)
  --(8649,-2388)--(8693,-2401)--(8735,-2415)--(8774,-2429)--(8810,-2443)--(8843,-2456)
  --(8873,-2470)--(8901,-2483)--(8926,-2497)--(8949,-2510)--(8969,-2523)--(8986,-2536)
  --(9002,-2549)--(9015,-2562)--(9027,-2574)--(9037,-2586)--(9045,-2599)--(9052,-2610)
  --(9058,-2622)--(9062,-2634)--(9066,-2645)--(9069,-2656)--(9071,-2667)--(9073,-2678)
  --(9074,-2689)--(9075,-2700)--(9077,-2721)--(9079,-2742)--(9080,-2763)--(9082,-2785)
  --(9083,-2807)--(9084,-2830)--(9085,-2853)--(9085,-2877)--(9086,-2901)--(9086,-2925)
  --(9086,-2949)--(9085,-2973)--(9085,-2997)--(9084,-3020)--(9083,-3043)--(9082,-3065)
  --(9080,-3087)--(9079,-3108)--(9077,-3129)--(9075,-3150)--(9074,-3161)--(9073,-3172)
  --(9071,-3183)--(9069,-3194)--(9066,-3205)--(9062,-3216)--(9058,-3228)--(9052,-3240)
  --(9045,-3251)--(9037,-3264)--(9027,-3276)--(9015,-3288)--(9002,-3301)--(8986,-3314)
  --(8969,-3327)--(8949,-3340)--(8926,-3353)--(8901,-3367)--(8873,-3380)--(8843,-3394)
  --(8810,-3407)--(8774,-3421)--(8735,-3435)--(8693,-3449)--(8649,-3462)--(8601,-3476)
  --(8550,-3490)--(8496,-3503)--(8439,-3517)--(8379,-3530)--(8316,-3544)--(8250,-3557)
  --(8180,-3571)--(8107,-3584)--(8031,-3597)--(7950,-3611)--(7866,-3624)--(7778,-3638)
  --(7733,-3644)--(7687,-3651)--(7640,-3658)--(7592,-3664)--(7542,-3671)--(7491,-3678)
  --(7438,-3685)--(7384,-3692)--(7329,-3699)--(7271,-3706)--(7212,-3714)--(7151,-3721)
  --(7088,-3729)--(7023,-3736)--(6956,-3744)--(6887,-3752)--(6816,-3760)--(6742,-3768)
  --(6666,-3777)--(6588,-3785)--(6507,-3794)--(6423,-3803)--(6337,-3812)--(6248,-3822)
  --(6157,-3831)--(6063,-3841)--(5966,-3851)--(5866,-3861)--(5764,-3872)--(5659,-3882)
  --(5551,-3893)--(5441,-3904)--(5328,-3916)--(5212,-3927)--(5094,-3939)--(4974,-3951)
  --(4851,-3963)--(4727,-3975)--(4600,-3987)--(4472,-4000)--(4343,-4012)--(4213,-4025)
  --(4082,-4038)--(3951,-4050)--(3819,-4063)--(3688,-4075)--(3558,-4088)--(3428,-4100)
  --(3301,-4112)--(3175,-4124)--(3052,-4136)--(2932,-4147)--(2815,-4159)--(2702,-4169)
  --(2593,-4180)--(2489,-4189)--(2389,-4199)--(2295,-4208)--(2207,-4216)--(2124,-4224)
  --(2048,-4231)--(1977,-4237)--(1913,-4243)--(1855,-4249)--(1804,-4254)--(1758,-4258)
  --(1719,-4262)--(1685,-4265)--(1656,-4267)--(1633,-4270)--(1615,-4271)--(1600,-4273)
  --(1590,-4274)--(1583,-4274)--(1578,-4275)--(1576,-4275)--(1575,-4275);
\filldraw (1350,-1575)--(1348,-1575)--(1343,-1575)--(1335,-1575)--(1321,-1575)--(1302,-1575)
  --(1277,-1575)--(1245,-1576)--(1207,-1576)--(1163,-1576)--(1113,-1576)--(1058,-1577)
  --(999,-1577)--(937,-1578)--(873,-1578)--(807,-1579)--(741,-1580)--(676,-1581)
  --(611,-1581)--(548,-1582)--(486,-1583)--(427,-1584)--(371,-1585)--(317,-1587)
  --(266,-1588)--(218,-1589)--(173,-1590)--(130,-1592)--(89,-1593)--(51,-1595)
  --(15,-1597)--(-18,-1599)--(-50,-1601)--(-81,-1603)--(-109,-1605)--(-137,-1607)
  --(-163,-1610)--(-188,-1613)--(-216,-1616)--(-243,-1619)--(-269,-1623)--(-294,-1628)
  --(-318,-1633)--(-341,-1638)--(-363,-1644)--(-384,-1651)--(-404,-1659)--(-424,-1667)
  --(-442,-1677)--(-460,-1687)--(-477,-1699)--(-492,-1711)--(-507,-1725)--(-521,-1739)
  --(-534,-1755)--(-546,-1773)--(-557,-1791)--(-567,-1810)--(-577,-1831)--(-585,-1853)
  --(-593,-1876)--(-600,-1900)--(-607,-1925)--(-612,-1952)--(-618,-1979)--(-622,-2008)
  --(-627,-2038)--(-631,-2070)--(-634,-2103)--(-638,-2138)--(-640,-2168)--(-643,-2200)
  --(-645,-2233)--(-647,-2267)--(-649,-2303)--(-651,-2340)--(-653,-2379)--(-655,-2419)
  --(-657,-2460)--(-658,-2503)--(-659,-2546)--(-661,-2591)--(-662,-2637)--(-663,-2684)
  --(-663,-2731)--(-664,-2779)--(-664,-2827)--(-665,-2876)--(-665,-2925)--(-665,-2974)
  --(-664,-3023)--(-664,-3071)--(-663,-3119)--(-663,-3166)--(-662,-3213)--(-661,-3259)
  --(-659,-3304)--(-658,-3347)--(-657,-3390)--(-655,-3431)--(-653,-3471)--(-651,-3510)
  --(-649,-3547)--(-647,-3583)--(-645,-3617)--(-643,-3650)--(-640,-3682)--(-638,-3713)
  --(-634,-3747)--(-631,-3780)--(-627,-3812)--(-622,-3842)--(-618,-3871)--(-612,-3898)
  --(-607,-3925)--(-600,-3950)--(-593,-3974)--(-585,-3997)--(-577,-4019)--(-567,-4040)
  --(-557,-4059)--(-546,-4077)--(-534,-4095)--(-521,-4111)--(-507,-4125)--(-492,-4139)
  --(-477,-4151)--(-460,-4163)--(-442,-4173)--(-424,-4183)--(-404,-4191)--(-384,-4199)
  --(-363,-4206)--(-341,-4212)--(-318,-4217)--(-294,-4222)--(-269,-4227)--(-243,-4231)
  --(-216,-4234)--(-188,-4238)--(-163,-4240)--(-137,-4243)--(-109,-4245)--(-81,-4247)
  --(-50,-4249)--(-18,-4251)--(15,-4253)--(51,-4255)--(89,-4257)--(130,-4258)
  --(173,-4260)--(218,-4261)--(266,-4262)--(317,-4263)--(371,-4265)--(427,-4266)
  --(486,-4267)--(548,-4268)--(611,-4269)--(676,-4269)--(741,-4270)--(807,-4271)
  --(873,-4272)--(937,-4272)--(999,-4273)--(1058,-4273)--(1113,-4274)--(1163,-4274)
  --(1207,-4274)--(1245,-4274)--(1277,-4275)--(1302,-4275)--(1321,-4275)--(1335,-4275)
  --(1343,-4275)--(1348,-4275)--(1350,-4275);
\pgfsetfillcolor{.!10}
\filldraw (1575,-1800)--(1576,-1800)--(1579,-1800)--(1584,-1801)--(1592,-1801)--(1603,-1802)
  --(1619,-1803)--(1640,-1805)--(1666,-1807)--(1697,-1809)--(1734,-1812)--(1777,-1816)
  --(1827,-1819)--(1883,-1824)--(1945,-1829)--(2013,-1834)--(2088,-1840)--(2169,-1846)
  --(2255,-1853)--(2347,-1860)--(2443,-1867)--(2544,-1875)--(2649,-1884)--(2758,-1892)
  --(2870,-1901)--(2984,-1910)--(3101,-1919)--(3219,-1929)--(3338,-1938)--(3457,-1948)
  --(3577,-1958)--(3697,-1967)--(3816,-1977)--(3934,-1987)--(4051,-1996)--(4166,-2006)
  --(4279,-2015)--(4390,-2024)--(4499,-2033)--(4606,-2042)--(4710,-2051)--(4811,-2060)
  --(4910,-2068)--(5006,-2077)--(5099,-2085)--(5189,-2093)--(5277,-2101)--(5362,-2108)
  --(5444,-2116)--(5523,-2123)--(5600,-2131)--(5674,-2138)--(5745,-2145)--(5814,-2151)
  --(5881,-2158)--(5946,-2165)--(6008,-2171)--(6068,-2178)--(6126,-2184)--(6182,-2190)
  --(6236,-2196)--(6288,-2202)--(6339,-2208)--(6388,-2214)--(6436,-2220)--(6482,-2226)
  --(6526,-2232)--(6570,-2238)--(6612,-2244)--(6653,-2250)--(6725,-2261)--(6794,-2272)
  --(6860,-2283)--(6923,-2294)--(6983,-2305)--(7040,-2316)--(7095,-2328)--(7147,-2339)
  --(7196,-2351)--(7243,-2363)--(7287,-2375)--(7329,-2387)--(7369,-2399)--(7406,-2412)
  --(7440,-2424)--(7472,-2437)--(7502,-2450)--(7530,-2462)--(7555,-2475)--(7578,-2488)
  --(7598,-2500)--(7617,-2513)--(7634,-2526)--(7649,-2538)--(7662,-2551)--(7673,-2563)
  --(7683,-2575)--(7691,-2587)--(7699,-2599)--(7704,-2611)--(7709,-2622)--(7713,-2634)
  --(7716,-2645)--(7719,-2656)--(7721,-2667)--(7723,-2678)--(7724,-2689)--(7725,-2700)
  --(7727,-2721)--(7729,-2742)--(7730,-2763)--(7732,-2785)--(7733,-2807)--(7734,-2830)
  --(7735,-2853)--(7735,-2877)--(7736,-2901)--(7736,-2925)--(7736,-2949)--(7735,-2973)
  --(7735,-2997)--(7734,-3020)--(7733,-3043)--(7732,-3065)--(7730,-3087)--(7729,-3108)
  --(7727,-3129)--(7725,-3150)--(7724,-3161)--(7723,-3172)--(7721,-3183)--(7719,-3194)
  --(7716,-3205)--(7713,-3216)--(7709,-3228)--(7704,-3239)--(7699,-3251)--(7691,-3263)
  --(7683,-3275)--(7673,-3287)--(7662,-3299)--(7649,-3312)--(7634,-3324)--(7617,-3337)
  --(7598,-3350)--(7578,-3362)--(7555,-3375)--(7530,-3388)--(7502,-3400)--(7472,-3413)
  --(7440,-3426)--(7406,-3438)--(7369,-3451)--(7329,-3463)--(7287,-3475)--(7243,-3487)
  --(7196,-3499)--(7147,-3511)--(7095,-3522)--(7040,-3534)--(6983,-3545)--(6923,-3556)
  --(6860,-3567)--(6794,-3578)--(6725,-3589)--(6653,-3600)--(6612,-3606)--(6570,-3612)
  --(6526,-3618)--(6482,-3624)--(6436,-3630)--(6388,-3636)--(6339,-3642)--(6288,-3648)
  --(6236,-3654)--(6182,-3660)--(6126,-3666)--(6068,-3672)--(6008,-3679)--(5946,-3685)
  --(5881,-3692)--(5814,-3699)--(5745,-3705)--(5674,-3712)--(5600,-3719)--(5523,-3727)
  --(5444,-3734)--(5362,-3742)--(5277,-3749)--(5189,-3757)--(5099,-3765)--(5006,-3773)
  --(4910,-3782)--(4811,-3790)--(4710,-3799)--(4606,-3808)--(4499,-3817)--(4390,-3826)
  --(4279,-3835)--(4166,-3844)--(4051,-3854)--(3934,-3863)--(3816,-3873)--(3697,-3883)
  --(3577,-3892)--(3457,-3902)--(3338,-3912)--(3219,-3921)--(3101,-3931)--(2984,-3940)
  --(2870,-3949)--(2758,-3958)--(2649,-3966)--(2544,-3975)--(2443,-3983)--(2347,-3990)
  --(2255,-3997)--(2169,-4004)--(2088,-4010)--(2013,-4016)--(1945,-4021)--(1883,-4026)
  --(1827,-4031)--(1777,-4034)--(1734,-4038)--(1697,-4041)--(1666,-4043)--(1640,-4045)
  --(1619,-4047)--(1603,-4048)--(1592,-4049)--(1584,-4049)--(1579,-4050)--(1576,-4050)
  --(1575,-4050);
\filldraw (1350,-1800)--(1347,-1800)--(1342,-1800)--(1331,-1800)--(1315,-1800)--(1293,-1800)
  --(1264,-1801)--(1228,-1801)--(1188,-1802)--(1142,-1802)--(1094,-1803)--(1042,-1804)
  --(990,-1804)--(937,-1805)--(886,-1806)--(835,-1808)--(787,-1809)--(741,-1810)
  --(698,-1812)--(658,-1813)--(620,-1815)--(585,-1817)--(552,-1819)--(521,-1821)
  --(493,-1823)--(467,-1826)--(442,-1828)--(418,-1831)--(396,-1834)--(375,-1838)
  --(353,-1841)--(333,-1845)--(313,-1850)--(294,-1855)--(275,-1861)--(257,-1868)
  --(240,-1875)--(223,-1883)--(208,-1892)--(192,-1902)--(178,-1913)--(164,-1926)
  --(151,-1939)--(138,-1954)--(127,-1970)--(116,-1987)--(106,-2006)--(97,-2025)
  --(88,-2046)--(80,-2068)--(73,-2091)--(66,-2115)--(60,-2141)--(55,-2167)--(50,-2195)
  --(45,-2225)--(41,-2255)--(38,-2288)--(35,-2315)--(32,-2345)--(29,-2375)--(27,-2407)
  --(24,-2440)--(22,-2475)--(20,-2511)--(18,-2548)--(17,-2586)--(15,-2626)--(14,-2666)
  --(13,-2708)--(12,-2750)--(11,-2793)--(11,-2837)--(10,-2881)--(10,-2925)--(10,-2969)
  --(11,-3013)--(11,-3057)--(12,-3100)--(13,-3142)--(14,-3184)--(15,-3224)--(17,-3264)
  --(18,-3302)--(20,-3339)--(22,-3375)--(24,-3410)--(27,-3443)--(29,-3475)--(32,-3505)
  --(35,-3535)--(38,-3563)--(41,-3595)--(45,-3625)--(50,-3655)--(55,-3683)--(60,-3709)
  --(66,-3735)--(73,-3759)--(80,-3782)--(88,-3804)--(97,-3825)--(106,-3844)--(116,-3863)
  --(127,-3880)--(138,-3896)--(151,-3911)--(164,-3924)--(178,-3937)--(192,-3948)
  --(208,-3958)--(223,-3967)--(240,-3975)--(257,-3982)--(275,-3989)--(294,-3995)
  --(313,-4000)--(333,-4005)--(353,-4009)--(375,-4013)--(396,-4016)--(418,-4019)
  --(442,-4022)--(467,-4024)--(493,-4027)--(521,-4029)--(552,-4031)--(585,-4033)
  --(620,-4035)--(658,-4037)--(698,-4038)--(741,-4040)--(787,-4041)--(835,-4042)
  --(886,-4044)--(937,-4045)--(990,-4046)--(1042,-4046)--(1094,-4047)--(1142,-4048)
  --(1188,-4048)--(1228,-4049)--(1264,-4049)--(1293,-4050)--(1315,-4050)--(1331,-4050)
  --(1342,-4050)--(1347,-4050)--(1350,-4050);
\pgfsetfillcolor{.!15}
\filldraw (1575,-2025)--(1577,-2025)--(1581,-2026)--(1589,-2026)--(1601,-2028)--(1619,-2029)
  --(1642,-2032)--(1672,-2035)--(1708,-2038)--(1751,-2042)--(1800,-2047)--(1856,-2053)
  --(1916,-2059)--(1982,-2066)--(2051,-2073)--(2124,-2080)--(2200,-2088)--(2276,-2096)
  --(2354,-2104)--(2432,-2113)--(2509,-2121)--(2585,-2130)--(2660,-2138)--(2732,-2146)
  --(2802,-2154)--(2870,-2162)--(2934,-2170)--(2996,-2177)--(3055,-2185)--(3112,-2192)
  --(3165,-2199)--(3216,-2206)--(3264,-2213)--(3309,-2220)--(3352,-2227)--(3393,-2233)
  --(3431,-2240)--(3468,-2247)--(3503,-2253)--(3536,-2260)--(3567,-2267)--(3597,-2274)
  --(3625,-2281)--(3653,-2288)--(3693,-2299)--(3731,-2310)--(3767,-2322)--(3801,-2334)
  --(3833,-2347)--(3863,-2360)--(3891,-2374)--(3918,-2388)--(3942,-2402)--(3965,-2417)
  --(3986,-2433)--(4005,-2448)--(4022,-2464)--(4038,-2480)--(4052,-2496)--(4064,-2513)
  --(4075,-2529)--(4085,-2545)--(4093,-2562)--(4099,-2578)--(4105,-2594)--(4110,-2609)
  --(4114,-2625)--(4117,-2640)--(4120,-2655)--(4122,-2670)--(4123,-2685)--(4125,-2700)
  --(4127,-2721)--(4129,-2742)--(4130,-2763)--(4132,-2785)--(4133,-2807)--(4134,-2830)
  --(4135,-2853)--(4135,-2877)--(4136,-2901)--(4136,-2925)--(4136,-2949)--(4135,-2973)
  --(4135,-2997)--(4134,-3020)--(4133,-3043)--(4132,-3065)--(4130,-3087)--(4129,-3108)
  --(4127,-3129)--(4125,-3150)--(4123,-3165)--(4122,-3180)--(4120,-3195)--(4117,-3210)
  --(4114,-3225)--(4110,-3241)--(4105,-3256)--(4099,-3272)--(4093,-3288)--(4085,-3305)
  --(4075,-3321)--(4064,-3337)--(4052,-3354)--(4038,-3370)--(4022,-3386)--(4005,-3402)
  --(3986,-3417)--(3965,-3433)--(3942,-3448)--(3918,-3462)--(3891,-3476)--(3863,-3490)
  --(3833,-3503)--(3801,-3516)--(3767,-3528)--(3731,-3540)--(3693,-3551)--(3653,-3563)
  --(3625,-3569)--(3597,-3576)--(3567,-3583)--(3536,-3590)--(3503,-3597)--(3468,-3603)
  --(3431,-3610)--(3393,-3617)--(3352,-3623)--(3309,-3630)--(3264,-3637)--(3216,-3644)
  --(3165,-3651)--(3112,-3658)--(3055,-3665)--(2996,-3673)--(2934,-3680)--(2870,-3688)
  --(2802,-3696)--(2732,-3704)--(2660,-3712)--(2585,-3720)--(2509,-3729)--(2432,-3737)
  --(2354,-3746)--(2276,-3754)--(2200,-3762)--(2124,-3770)--(2051,-3777)--(1982,-3784)
  --(1916,-3791)--(1856,-3797)--(1800,-3803)--(1751,-3808)--(1708,-3812)--(1672,-3815)
  --(1642,-3818)--(1619,-3821)--(1601,-3822)--(1589,-3824)--(1581,-3824)--(1577,-3825)
  --(1575,-3825);
\filldraw (1350,-2025)--(1347,-2025)--(1339,-2025)--(1327,-2025)--(1308,-2026)--(1284,-2027)
  --(1256,-2027)--(1225,-2029)--(1192,-2030)--(1160,-2031)--(1129,-2033)--(1100,-2035)
  --(1074,-2037)--(1049,-2040)--(1027,-2043)--(1006,-2046)--(987,-2049)--(970,-2053)
  --(953,-2058)--(938,-2063)--(924,-2067)--(910,-2073)--(897,-2079)--(884,-2086)
  --(871,-2094)--(858,-2103)--(845,-2114)--(833,-2125)--(821,-2138)--(809,-2153)
  --(798,-2169)--(787,-2186)--(777,-2205)--(767,-2225)--(758,-2247)--(750,-2270)
  --(742,-2294)--(735,-2320)--(729,-2347)--(723,-2376)--(717,-2406)--(713,-2438)
  --(709,-2462)--(706,-2488)--(703,-2515)--(700,-2543)--(698,-2573)--(696,-2604)
  --(694,-2636)--(692,-2669)--(690,-2703)--(689,-2739)--(687,-2775)--(687,-2812)
  --(686,-2849)--(686,-2887)--(685,-2925)--(686,-2963)--(686,-3001)--(687,-3038)
  --(687,-3075)--(689,-3111)--(690,-3147)--(692,-3181)--(694,-3214)--(696,-3246)
  --(698,-3277)--(700,-3307)--(703,-3335)--(706,-3362)--(709,-3388)--(713,-3413)
  --(717,-3444)--(723,-3474)--(729,-3503)--(735,-3530)--(742,-3556)--(750,-3580)
  --(758,-3603)--(767,-3625)--(777,-3645)--(787,-3664)--(798,-3681)--(809,-3697)
  --(821,-3712)--(833,-3725)--(845,-3736)--(858,-3747)--(871,-3756)--(884,-3764)
  --(897,-3771)--(910,-3777)--(924,-3783)--(938,-3788)--(953,-3792)--(970,-3797)
  --(987,-3801)--(1006,-3804)--(1027,-3807)--(1049,-3810)--(1074,-3813)--(1100,-3815)
  --(1129,-3817)--(1160,-3819)--(1192,-3820)--(1225,-3821)--(1256,-3823)--(1284,-3823)
  --(1308,-3824)--(1327,-3825)--(1339,-3825)--(1347,-3825)--(1350,-3825);
\pgfsetlinewidth{+15\XFigu}
\draw (1800,-2925)--(9675,-2925);
\draw (2250,-2835)--(9000,-2835);
\draw (2475,-2745)--(4050,-2745);
\draw (2700,-2655)--(3825,-2655);
\draw (4500,-2745)--(6075,-2745);
\draw (4950,-2655)--(6750,-2655);
\draw (6525,-2745)--(7650,-2745);
\draw (5400,-2565)--(7425,-2565);
\draw (8100,-2745)--(8550,-2745);
\pgfsetlinewidth{+45\XFigu}
\pgfsetdash{{+150\XFigu}{+75\XFigu}{+15\XFigu}{+75\XFigu}}{+0pt}
\draw (4275,-2610)--(4275,-3375);
\draw (2025,-2610)--(2025,-3330);
\draw (7875,-2610)--(7875,-3375);
\draw (9225,-2610)--(9225,-3375);
\pgfsetfillcolor{.}
\pgftext[base,left,at=\pgfqpointxy{1800}{-2475}] {\fontsize{7}{8.4}\usefont{T1}{ptm}{m}{n}$B_{i-2}$};
\pgftext[base,left,at=\pgfqpointxy{4050}{-2475}] {\fontsize{7}{8.4}\usefont{T1}{ptm}{m}{n}$B_{i-1}$};
\pgftext[base,left,at=\pgfqpointxy{7650}{-2475}] {\fontsize{7}{8.4}\usefont{T1}{ptm}{m}{n}$B_{i}$};
\pgftext[base,left,at=\pgfqpointxy{9000}{-2475}] {\fontsize{7}{8.4}\usefont{T1}{ptm}{m}{n}$B_{i+1}$};
\pgftext[base,left,at=\pgfqpointxy{2925}{-3240}] {\fontsize{10}{12}\usefont{T1}{ptm}{m}{n}$\partial V_{i-1}$};
\pgftext[base,left,at=\pgfqpointxy{5310}{-3240}] {\fontsize{10}{12}\usefont{T1}{ptm}{m}{n}$\partial V_{i}$};
\pgftext[base,left,at=\pgfqpointxy{8100}{-3285}] {\fontsize{10}{12}\usefont{T1}{ptm}{m}{n}$\partial V_{i+1}$};
\pgftext[base,left,at=\pgfqpointxy{810}{-2700}] {\fontsize{8}{9.6}\usefont{T1}{ptm}{m}{n}$V_{i-1}$};
\pgftext[base,left,at=\pgfqpointxy{180}{-2970}] {\fontsize{8}{9.6}\usefont{T1}{ptm}{m}{n}$V_{i}$};
\pgftext[base,left,at=\pgfqpointxy{-540}{-3240}] {\fontsize{8}{9.6}\usefont{T1}{ptm}{m}{n}$V_{i+1}$};
\endtikzpicture%
\caption{Sets $B_i$, $V_i$ and $\partial V_i$ on the interval
  representation of $G$ (for instance, $B_i$ contains all the
  intervals crossing the corresponding dotted line, $\partial V_i$
  contains all the intervals intersecting the zone between $B_{i-1}$
  and $B_i$, and $V_i$ contains all the intervals intersecting the
  corresponding zone).}\label{fig:int-rep}
\end{center}  
\end{figure}

\begin{claim}\label{claim2}
For every $i\in [k]$, the graph $\partial G_i$ is either a clique of
order less than $t$ or a $t$-connected graph.
\end{claim}
\begin{proof}[Proof of Claim \ref{claim2}]
Let $i\in [k]$.
By definition, 
there are indices $i_1$ and $i_2$ with $i_1<i_2$ such that 
$\partial V_i=\bigcup\limits_{j=i_1}^{i_2}C_j$.
Now, 
either $c_j<t$ for every index $j$ with $i_1\leq j\leq i_2$, 
which implies that there is an index $\ell$ with $i_1<\ell<i_2$ and
 $c_{i_1}<\ldots<c_{\ell-1}<c_{\ell}>c_{\ell+1}>\ldots>c_{i_2}$,
in which case $\partial G_i$ is a clique of order $c_\ell<t$;
or there are indices $i_1'$ and $i_2'$ with $i_1<i_1'\leq i_2'<i_2$
such that 
$c_{i_1}<c_{i_1+1}<\ldots<c_{i'_1}$,
$c_j\geq t$ for every index $j$ with $i'_1\leq j\leq i_2'$,
and 
$c_{i_2'}>c_{i_2'+1}>\ldots>c_{i_2}$, 
in which case Claim \ref{claim1} implies that $\partial G_i$ 
is $t$-connected.
\end{proof}
As explained above, 
we apply dynamic programming
calculating partial information for each $G_i$.
This information should be rich enough 
to capture the influence on $G_i$ from outside of $G_i$
of all possible cascades of 
a minimum dynamic monopoly $D$ of $(G,\tau)$.
Since the only vertices of $G_i$ with neighbors outside of $G_i$
are in $B_i$, 
this leads us to considering 
a localized version of a cascade
that specifies
(i) all possible intersections of $D$ with $B_i$,
(ii) all possible orders, 
in which the elements of $B_i$ appear in a cascade, and
(iii) all possible amounts of help that each vertex in $B_i$ 
receives from outside of $G_i$ when it enters the hull of $D$.
Consequently, for every $i\in [k]$,
a {\it local cascade} for $G_i$ 
is defined as a triple $(X_i,\prec_i,\rho_i)$, where 
\begin{enumerate}[(i)]
\item $X_i$ is a subset of $B_i$,
\item $\prec_i$ is a linear order on $B_i$ such that $u\prec_i v$ 
for every $u\in X_i$ and every $v\in B_i\setminus X_i$, and 
\item $\rho_i:B_i\setminus X\to \{ 0,1,\ldots,n\}$.
\end{enumerate}
Since $|B_i|\leq t-1$,
there are $O\Big( 2^{t-1}(t-1)!(n+1)^{t-1}\Big)$ local cascades for $G_i$.

For each local cascade for $G_i$, 
we are interested in the minimum
number of vertices from $V_i\setminus B_i$ 
that need to be added to $X_i$ 
in order to obtain the intersection with $V_i$ 
of some dynamic monopoly that is compatible with the local cascade.
More precisely, for a local cascade $(X_i,\prec_i,\rho_i)$ for $G_i$,
let ${\rm dyn}_i(X_i,\prec_i,\rho_i)$ 
be the minimum order of a subset $Y_i$ of $V_i\setminus B_i$ 
such that the following conditions hold:
\begin{enumerate}
\item[(iv)] $|(X_i\cup Y_i)\cap \partial V_j|\leq t$ for every $j\in [i]$.
\item[(v)] There is a linear extension 
$u_1\prec \ldots \prec u_{n(G_i)}$
of $\prec_i$ to $V(G_i)$ such that
$u\prec v$
for every $u\in X_i\cup Y_i$
and every $v\in V_i\setminus (X_i\cup Y_i)$, and, 
for every $j$ in $[n(G_i)]$,
\begin{enumerate}[(a)]
\item either $u_j\in X_i\cup Y_i$,
\item or $u_j\not\in Y_i\cup B_i$ and 
$\Big|N_G(u_j)\cap \{ u_1,\ldots,u_{j-1}\}\Big|\geq \tau(u_j)$,
\item or $u_j\in B_i\setminus X_i$ and 
$\Big|N_G(u_j)\cap \{ u_1,\ldots,u_{j-1}\}\Big|\geq \tau(u_j)-\rho(u_j)$.
\end{enumerate}
\end{enumerate}
If no such set $Y_i$ exists, then ${\rm dyn}_i(X_i,\prec_i,\rho_i)=\infty$.
Note that (a) and (b) are as in the definition of a cascade,
and that (c) incorporates the assumption that $u_j$ 
has $\rho(u_j)$ neighbors outside of $G_i$ 
when it enters the hull.

By definition, we have $G=G_k$, and $|B_k|=1$,
which implies that there are exactly two local cascades $(X_k,\prec_k,\rho_k)$
for $G_k$ with $\rho_k(u)=0$ for every $u\in B_k\setminus X_k$;
these are the local cascades $(B_k,\emptyset,0)$ and $(\emptyset,\emptyset,0)$.

\begin{claim}\label{claim3}
${\rm dyn}(G,\tau)=
\min\Big\{ 1+{\rm dyn}_k(B_k,\emptyset,0),
0+{\rm dyn}_k(\emptyset,\emptyset,0)\Big\}.$
\end{claim}
\begin{proof}[Proof of Claim \ref{claim3}]
Let $D$ be a dynamic monopoly of $(G,\tau)$ 
of order ${\rm dyn}(G,\tau)$.  

Our first goal is to show that (iv) holds for 
$i=k$,
$X_k=D\cap B_k$, and
$Y_k=D\setminus X_k$.  
Suppose, for contradiction, that 
$|D\cap \partial V_j|> t$ for some $j\in [k]$.
Clearly, $\partial G_j$ can not be a clique of size less than $t$ in this case.
Therefore, by Claim~\ref{claim2},
$\partial G_j$ is $t$-connected, and, 
by Theorem~\ref{theoremchordal},
there is a dynamic monopoly $D_j$ of $(\partial G_j,\tau)$ of size at most $t$.
Now, $(D\setminus \partial V_j)\cup D_j$ 
is a dynamic monopoly of $(G,\tau)$ 
of order less than $D$, which is a contradiction.
Hence, (iv) holds. 

Let $u_1\prec \cdots \prec u_n$ be a cascade for $D$.
Since this cascade is a linear extension 
of the trivial linear order on the one-element set $B_k$,
we obtain (v) with $\rho_k(u)=0$ for every $u\in B_k\setminus X_k$.
This implies 
$|X_k|+{\rm dyn}_k(X_k,\emptyset,0)\leq |X_k|+|Y_k|={\rm dyn}(G,\tau)$.

Conversely, let $X_k\subseteq B_k$ be such that 
$\min\Big\{ 1+{\rm dyn}_k(B_k,\emptyset,0),
0+{\rm dyn}_k(\emptyset,\emptyset,0)\Big\}$
equals $|X_k|+{\rm dyn}_k(X_k,\emptyset,0)$.
If $Y_k$ is as in the definition of ${\rm dyn}_k(X_k,\emptyset,0)$,
then (v) and $\rho_k=0$ imply that $X_k\cup Y_k$ is a dynamic monopoly of $(G,\tau)$, 
which implies
${\rm dyn}(G,\tau)\leq |X_k|+|Y_k|=|X_k|+{\rm dyn}_k(X_k,\emptyset,0)$.
\end{proof}
Our next two claims imply that 
the values ${\rm dyn}_i(X_i,\prec_i,\rho_i)$
can be determined recursively in polynomial time.

\begin{claim}\label{claim4}
For every local cascade $(X_1,\prec_1,\rho_1)$ for $G_1$,
the value ${\rm dyn}_1(X_1,\prec_1,\rho_1)$
can be computed in polynomial time.
\end{claim}
\begin{proof}[Proof of Claim \ref{claim4}]
Let $v_1\prec_1 \ldots \prec_1 v_p$ be the linear order $\prec_1$ on $B_1$.
Since $V_1=\partial V_1$, 
every subset $Y_1$ of $V_1\setminus B_1$ 
satisfying condition (iv) has at most $t-|X_1|$ elements,
which implies that there are only $O(n^t)$ candidates for $Y_1$.
For each such set $Y_1$,
condition (v) holds if and only if
\begin{itemize}
\item[(b$'$)] $B_1\cup Y_1$ is a dynamic monopoly of $(G_1,\tau)$, and
\item[(c$'$)] for every $i$ in $[p]$ with $v_i\in B_1\setminus X_1$, 
the hull of the set
$$\Big\{ v_j:j\in [i-1]\Big\}
\cup X_1\cup Y_1$$
in $\Big(G_1-\Big\{ v_j:j\in [p]\setminus [i-1]\Big\},\tau\Big)$
contains at least $\tau(v_i)-\rho(v_i)$ many neighbors of $v_i$.
\end{itemize}
In fact, if there is a linear extension $u_1\prec \ldots \prec u_{n(G_i)}$ of $\prec_1$ satisfying (v), then (a) and (b) imply (b$'$), and (c) implies (c$'$).
Conversely, if (b$'$) and (c$'$) hold, 
then concatenating cascades for the $p$ hulls considered in (c$'$)
for $i$ from $1$ up to $p$,
and removing all but the first appearance of each vertex 
in the resulting sequence,
yields a linear order satisfying (v).
Since (b$'$) and (c$'$) can be checked efficiently
for the polynomially many candidates for $Y_1$,
the claim follows.
\end{proof}

\begin{claim}\label{claim5}
For every $i\in [k]\setminus \{ 1\}$ 
and every local cascade $(X_i,\prec_i,\rho_i)$ for $G_i$,
given the values ${\rm dyn}_{i-1}(X_{i-1},\prec_{i-1}, \rho_{i-1})$
for all local cascades $(X_{i-1},\prec_{i-1},\rho_{i-1})$ for $G_{i-1}$,
the value ${\rm dyn}_i(X_i,\prec_i,\rho_i)$
can be computed in polynomial time.
\end{claim}
\begin{proof}[Proof of Claim \ref{claim5}]
By definition, we have $B_i\cap V_{i-1}\subseteq B_{i-1}$.  Therefore,
the two sets $B_{i-1}'=B_i\cap V_{i-1}$ and
$B''_{i-1}=B_{i-1}\setminus B'_{i-1}$ partition the set $B_{i-1}$.
Let $X_{i-1}'=X_i\cap B_{i-1}$.  Note that $B_{i-1}'=B_i\cap B_{i-1}$,
$X_{i-1}'\subseteq B_{i-1}'$, and $B''_{i-1}=B_{i-1}\setminus B_i$,
cf. Figure~\ref{fig:set-Xi}.

\begin{figure}[t]
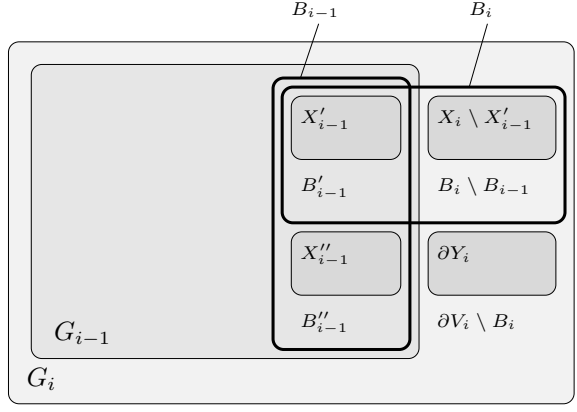

\begin{center}
\ifx\XFigwidth\undefined\dimen1=0pt\else\dimen1\XFigwidth\fi
\divide\dimen1 by 5604
\ifx\XFigheight\undefined\dimen3=0pt\else\dimen3\XFigheight\fi
\divide\dimen3 by 4041
\ifdim\dimen1=0pt\ifdim\dimen3=0pt\dimen1=2486sp\dimen3\dimen1
  \else\dimen1\dimen3\fi\else\ifdim\dimen3=0pt\dimen3\dimen1\fi\fi
\tikzpicture[x=+\dimen1, y=+\dimen3]
{\ifx\XFigu\undefined\catcode`\@11
\def\temp{\alloc@1\dimen\dimendef\insc@unt}\temp\XFigu\catcode`\@12\fi}
\XFigu2486sp
\ifdim\XFigu<0pt\XFigu-\XFigu\fi
\definecolor{xfigc32}{rgb}{0.612,0.000,0.000}
\definecolor{xfigc33}{rgb}{0.549,0.549,0.549}
\definecolor{xfigc34}{rgb}{0.549,0.549,0.549}
\definecolor{xfigc35}{rgb}{0.259,0.259,0.259}
\definecolor{xfigc36}{rgb}{0.549,0.549,0.549}
\definecolor{xfigc37}{rgb}{0.259,0.259,0.259}
\definecolor{xfigc38}{rgb}{0.549,0.549,0.549}
\definecolor{xfigc39}{rgb}{0.259,0.259,0.259}
\definecolor{xfigc40}{rgb}{0.549,0.549,0.549}
\definecolor{xfigc41}{rgb}{0.259,0.259,0.259}
\definecolor{xfigc42}{rgb}{0.549,0.549,0.549}
\definecolor{xfigc43}{rgb}{0.259,0.259,0.259}
\definecolor{xfigc44}{rgb}{0.557,0.557,0.557}
\definecolor{xfigc45}{rgb}{0.761,0.761,0.761}
\definecolor{xfigc46}{rgb}{0.431,0.431,0.431}
\definecolor{xfigc47}{rgb}{0.267,0.267,0.267}
\definecolor{xfigc48}{rgb}{0.557,0.561,0.557}
\definecolor{xfigc49}{rgb}{0.443,0.443,0.443}
\definecolor{xfigc50}{rgb}{0.682,0.682,0.682}
\definecolor{xfigc51}{rgb}{0.200,0.200,0.200}
\definecolor{xfigc52}{rgb}{0.580,0.576,0.584}
\definecolor{xfigc53}{rgb}{0.455,0.439,0.459}
\definecolor{xfigc54}{rgb}{0.333,0.333,0.333}
\definecolor{xfigc55}{rgb}{0.702,0.702,0.702}
\definecolor{xfigc56}{rgb}{0.765,0.765,0.765}
\definecolor{xfigc57}{rgb}{0.427,0.427,0.427}
\definecolor{xfigc58}{rgb}{0.271,0.271,0.271}
\definecolor{xfigc59}{rgb}{0.886,0.886,0.933}
\definecolor{xfigc60}{rgb}{0.580,0.580,0.604}
\definecolor{xfigc61}{rgb}{0.859,0.859,0.859}
\definecolor{xfigc62}{rgb}{0.631,0.631,0.718}
\definecolor{xfigc63}{rgb}{0.929,0.929,0.929}
\definecolor{xfigc64}{rgb}{0.878,0.878,0.878}
\definecolor{xfigc65}{rgb}{0.525,0.675,1.000}
\definecolor{xfigc66}{rgb}{0.439,0.439,1.000}
\definecolor{xfigc67}{rgb}{0.776,0.718,0.592}
\definecolor{xfigc68}{rgb}{0.937,0.973,1.000}
\definecolor{xfigc69}{rgb}{0.863,0.796,0.651}
\definecolor{xfigc70}{rgb}{0.251,0.251,0.251}
\definecolor{xfigc71}{rgb}{0.502,0.502,0.502}
\definecolor{xfigc72}{rgb}{0.753,0.753,0.753}
\definecolor{xfigc73}{rgb}{0.667,0.667,0.667}
\definecolor{xfigc74}{rgb}{0.780,0.765,0.780}
\definecolor{xfigc75}{rgb}{0.337,0.318,0.318}
\definecolor{xfigc76}{rgb}{0.843,0.843,0.843}
\definecolor{xfigc77}{rgb}{0.522,0.502,0.490}
\definecolor{xfigc78}{rgb}{0.824,0.824,0.824}
\definecolor{xfigc79}{rgb}{0.227,0.227,0.227}
\definecolor{xfigc80}{rgb}{0.271,0.451,0.667}
\definecolor{xfigc81}{rgb}{0.482,0.475,0.647}
\definecolor{xfigc82}{rgb}{0.451,0.459,0.549}
\definecolor{xfigc83}{rgb}{0.969,0.969,0.969}
\definecolor{xfigc84}{rgb}{0.255,0.271,0.255}
\definecolor{xfigc85}{rgb}{0.388,0.365,0.808}
\definecolor{xfigc86}{rgb}{0.745,0.745,0.745}
\definecolor{xfigc87}{rgb}{0.318,0.318,0.318}
\definecolor{xfigc88}{rgb}{0.906,0.890,0.906}
\definecolor{xfigc89}{rgb}{0.000,0.000,0.286}
\definecolor{xfigc90}{rgb}{0.475,0.475,0.475}
\definecolor{xfigc91}{rgb}{0.188,0.204,0.188}
\definecolor{xfigc92}{rgb}{0.255,0.255,0.255}
\definecolor{xfigc93}{rgb}{0.780,0.714,0.588}
\definecolor{xfigc94}{rgb}{0.867,0.616,0.576}
\definecolor{xfigc95}{rgb}{0.945,0.925,0.878}
\definecolor{xfigc96}{rgb}{0.886,0.784,0.659}
\definecolor{xfigc97}{rgb}{0.882,0.882,0.882}
\definecolor{xfigc98}{rgb}{0.855,0.478,0.102}
\definecolor{xfigc99}{rgb}{0.945,0.894,0.102}
\definecolor{xfigc100}{rgb}{0.533,0.490,0.761}
\definecolor{xfigc101}{rgb}{0.690,0.631,0.576}
\definecolor{xfigc102}{rgb}{0.514,0.486,0.867}
\definecolor{xfigc103}{rgb}{0.839,0.839,0.839}
\definecolor{xfigc104}{rgb}{0.549,0.549,0.647}
\definecolor{xfigc105}{rgb}{0.290,0.290,0.290}
\definecolor{xfigc106}{rgb}{0.549,0.420,0.420}
\definecolor{xfigc107}{rgb}{0.353,0.353,0.353}
\definecolor{xfigc108}{rgb}{0.388,0.388,0.388}
\definecolor{xfigc109}{rgb}{0.718,0.608,0.451}
\definecolor{xfigc110}{rgb}{0.255,0.576,1.000}
\definecolor{xfigc111}{rgb}{0.749,0.439,0.231}
\definecolor{xfigc112}{rgb}{0.859,0.467,0.000}
\definecolor{xfigc113}{rgb}{0.855,0.722,0.000}
\definecolor{xfigc114}{rgb}{0.000,0.392,0.000}
\definecolor{xfigc115}{rgb}{0.353,0.420,0.231}
\definecolor{xfigc116}{rgb}{0.827,0.827,0.827}
\definecolor{xfigc117}{rgb}{0.557,0.557,0.643}
\definecolor{xfigc118}{rgb}{0.953,0.725,0.365}
\definecolor{xfigc119}{rgb}{0.537,0.600,0.420}
\definecolor{xfigc120}{rgb}{0.392,0.392,0.392}
\definecolor{xfigc121}{rgb}{0.718,0.902,1.000}
\definecolor{xfigc122}{rgb}{0.525,0.753,0.925}
\definecolor{xfigc123}{rgb}{0.741,0.741,0.741}
\definecolor{xfigc124}{rgb}{0.827,0.584,0.322}
\definecolor{xfigc125}{rgb}{0.596,0.824,0.996}
\definecolor{xfigc126}{rgb}{0.380,0.380,0.380}
\definecolor{xfigc127}{rgb}{0.682,0.698,0.682}
\definecolor{xfigc128}{rgb}{1.000,0.604,0.000}
\definecolor{xfigc129}{rgb}{0.549,0.612,0.420}
\definecolor{xfigc130}{rgb}{0.969,0.420,0.000}
\definecolor{xfigc131}{rgb}{0.353,0.420,0.224}
\definecolor{xfigc132}{rgb}{0.549,0.612,0.420}
\definecolor{xfigc133}{rgb}{0.549,0.612,0.482}
\definecolor{xfigc134}{rgb}{0.094,0.290,0.094}
\definecolor{xfigc135}{rgb}{0.678,0.678,0.678}
\definecolor{xfigc136}{rgb}{0.969,0.741,0.353}
\definecolor{xfigc137}{rgb}{0.388,0.420,0.612}
\definecolor{xfigc138}{rgb}{0.871,0.000,0.000}
\definecolor{xfigc139}{rgb}{0.678,0.678,0.678}
\definecolor{xfigc140}{rgb}{0.969,0.741,0.353}
\definecolor{xfigc141}{rgb}{0.678,0.678,0.678}
\definecolor{xfigc142}{rgb}{0.969,0.741,0.353}
\definecolor{xfigc143}{rgb}{0.388,0.420,0.612}
\definecolor{xfigc144}{rgb}{0.322,0.420,0.161}
\definecolor{xfigc145}{rgb}{0.580,0.580,0.580}
\definecolor{xfigc146}{rgb}{0.000,0.388,0.000}
\definecolor{xfigc147}{rgb}{0.000,0.388,0.290}
\definecolor{xfigc148}{rgb}{0.482,0.518,0.290}
\definecolor{xfigc149}{rgb}{0.906,0.741,0.482}
\definecolor{xfigc150}{rgb}{0.647,0.710,0.776}
\definecolor{xfigc151}{rgb}{0.420,0.420,0.580}
\definecolor{xfigc152}{rgb}{0.518,0.420,0.420}
\definecolor{xfigc153}{rgb}{0.322,0.612,0.290}
\definecolor{xfigc154}{rgb}{0.839,0.906,0.906}
\definecolor{xfigc155}{rgb}{0.322,0.388,0.388}
\definecolor{xfigc156}{rgb}{0.094,0.420,0.290}
\definecolor{xfigc157}{rgb}{0.612,0.647,0.710}
\definecolor{xfigc158}{rgb}{1.000,0.580,0.000}
\definecolor{xfigc159}{rgb}{1.000,0.580,0.000}
\definecolor{xfigc160}{rgb}{0.000,0.388,0.290}
\definecolor{xfigc161}{rgb}{0.482,0.518,0.290}
\definecolor{xfigc162}{rgb}{0.388,0.451,0.482}
\definecolor{xfigc163}{rgb}{0.906,0.741,0.482}
\definecolor{xfigc164}{rgb}{0.094,0.290,0.094}
\definecolor{xfigc165}{rgb}{0.969,0.741,0.353}
\definecolor{xfigc166}{rgb}{0.000,0.000,0.000}
\definecolor{xfigc167}{rgb}{0.969,0.220,0.161}
\definecolor{xfigc168}{rgb}{0.000,0.000,0.000}
\definecolor{xfigc169}{rgb}{1.000,1.000,0.322}
\definecolor{xfigc170}{rgb}{0.322,0.475,0.290}
\definecolor{xfigc171}{rgb}{0.388,0.604,0.353}
\definecolor{xfigc172}{rgb}{0.776,0.380,0.259}
\definecolor{xfigc173}{rgb}{0.906,0.412,0.259}
\definecolor{xfigc174}{rgb}{1.000,0.475,0.322}
\definecolor{xfigc175}{rgb}{0.871,0.871,0.871}
\definecolor{xfigc176}{rgb}{0.953,0.933,0.827}
\definecolor{xfigc177}{rgb}{0.961,0.682,0.365}
\definecolor{xfigc178}{rgb}{0.584,0.808,0.600}
\definecolor{xfigc179}{rgb}{0.710,0.082,0.490}
\definecolor{xfigc180}{rgb}{0.933,0.933,0.933}
\definecolor{xfigc181}{rgb}{0.518,0.518,0.518}
\definecolor{xfigc182}{rgb}{0.482,0.482,0.482}
\definecolor{xfigc183}{rgb}{0.000,0.353,0.000}
\definecolor{xfigc184}{rgb}{0.906,0.451,0.451}
\definecolor{xfigc185}{rgb}{1.000,0.796,0.192}
\definecolor{xfigc186}{rgb}{0.161,0.475,0.290}
\definecolor{xfigc187}{rgb}{0.871,0.157,0.129}
\definecolor{xfigc188}{rgb}{0.129,0.349,0.776}
\definecolor{xfigc189}{rgb}{0.973,0.973,0.973}
\definecolor{xfigc190}{rgb}{0.902,0.902,0.902}
\definecolor{xfigc191}{rgb}{0.129,0.518,0.353}
\definecolor{xfigc192}{rgb}{0.906,0.906,0.906}
\definecolor{xfigc193}{rgb}{0.443,0.459,0.443}
\definecolor{xfigc194}{rgb}{0.851,0.851,0.851}
\definecolor{xfigc195}{rgb}{0.337,0.620,0.690}
\definecolor{xfigc196}{rgb}{0.788,0.788,0.788}
\definecolor{xfigc197}{rgb}{0.875,0.847,0.875}
\definecolor{xfigc198}{rgb}{0.969,0.953,0.969}
\definecolor{xfigc199}{rgb}{0.800,0.800,0.800}
\clip(888,-4242) rectangle (6492,-201);
\tikzset{inner sep=+0pt, outer sep=+0pt}
\pgfsetlinewidth{+7.5\XFigu}
\pgfsetfillcolor{.!5}
\filldraw (6480,-4230) [rounded corners=+105\XFigu] rectangle (900,-630);
\pgfsetfillcolor{.!10}
\filldraw (4950,-3780) [rounded corners=+105\XFigu] rectangle (1125,-855);
\pgfsetfillcolor{.}
\pgftext[base,left,at=\pgfqpointxy{3690}{-360}] {\fontsize{7}{8.4}\usefont{T1}{ptm}{m}{n}$B_{i-1}$};
\pgftext[base,left,at=\pgfqpointxy{1080}{-4050}] {\fontsize{10}{12}\usefont{T1}{ptm}{m}{n}$G_{i}$};
\pgftext[base,left,at=\pgfqpointxy{1350}{-3600}] {\fontsize{10}{12}\usefont{T1}{ptm}{m}{n}$G_{i-1}$};
\pgftext[base,left,at=\pgfqpointxy{5445}{-360}] {\fontsize{7}{8.4}\usefont{T1}{ptm}{m}{n}$B_{i}$};
\pgfsetlinewidth{+30\XFigu}
\draw (3510,-3690) [rounded corners=+105\XFigu] rectangle (4860,-990);
\pgfsetlinewidth{+7.5\XFigu}
\draw (3780,-990)--(3960,-450);
\pgfsetlinewidth{+30\XFigu}
\draw (6390,-2430) [rounded corners=+105\XFigu] rectangle (3600,-1080);
\pgfsetlinewidth{+7.5\XFigu}
\draw (5445,-1080)--(5625,-450);
\pgfsetfillcolor{.!15}
\filldraw (4770,-1800) [rounded corners=+105\XFigu] rectangle (3690,-1170);
\filldraw (4770,-3150) [rounded corners=+105\XFigu] rectangle (3690,-2520);
\filldraw (6300,-1800) [rounded corners=+105\XFigu] rectangle (5040,-1170);
\filldraw (6300,-3150) [rounded corners=+105\XFigu] rectangle (5040,-2520);
\pgfsetfillcolor{.}
\pgftext[base,left,at=\pgfqpointxy{3780}{-1440}] {\fontsize{7}{8.4}\usefont{T1}{ptm}{m}{n}$X'_{i-1}$};
\pgftext[base,left,at=\pgfqpointxy{3780}{-2790}] {\fontsize{7}{8.4}\usefont{T1}{ptm}{m}{n}$X''_{i-1}$};
\pgftext[base,left,at=\pgfqpointxy{5130}{-1440}] {\fontsize{7}{8.4}\usefont{T1}{ptm}{m}{n}$X_i\setminus X'_{i-1}$};
\pgftext[base,left,at=\pgfqpointxy{5130}{-2790}] {\fontsize{7}{8.4}\usefont{T1}{ptm}{m}{n}$\partial Y_i$};
\pgftext[base,left,at=\pgfqpointxy{3780}{-3465}] {\fontsize{7}{8.4}\usefont{T1}{ptm}{m}{n}$B''_{i-1}$};
\pgftext[base,left,at=\pgfqpointxy{5130}{-2115}] {\fontsize{7}{8.4}\usefont{T1}{ptm}{m}{n}$B_i\setminus B_{i-1}$};
\pgftext[base,left,at=\pgfqpointxy{3780}{-2115}] {\fontsize{7}{8.4}\usefont{T1}{ptm}{m}{n}$B'_{i-1}$};
\pgftext[base,left,at=\pgfqpointxy{5130}{-3465}] {\fontsize{7}{8.4}\usefont{T1}{ptm}{m}{n}$\partial V_i\setminus B_{i}$};
\endtikzpicture%
\caption{$G_i$ and relevant subsets of $V_i$.}\label{fig:set-Xi}
\end{center}  
\end{figure}
Our approach to determine 
${\rm dyn}_i(X_i,\prec_i,\rho_i)$
relies on considering all candidates for the two intersections 
--- later referred to as $X''_{i-1}$ and $\partial Y_i$ ---
of a set $Y_i$ as in the definition of ${\rm dyn}_i(X_i,\prec_i,\rho_i)$ 
with the two sets $B''_{i-1}$ 
and 
$\partial V_i\setminus (B_i\cup B_{i-1})$.
By (iv), 
these two intersections may contain a total of at most $t-|X_i|$ vertices.
In order to exploit the given values 
${\rm dyn}_{i-1}(X_{i-1},\prec_{i-1},\rho_{i-1})$,
we decouple $\partial G_i$ from $G_i-B_i$,
which leads us to consider 
all candidates for an extension $\prec_{(i-1,i)}$ of $\prec_i$ 
to $B_{i-1}\cup B_i$
specifying a possible order 
in which the vertices in $B_{i-1}\cup B_i$
appear in a cascade.
Fixing the triple $\Big(X''_{i-1},\partial Y_i,\prec_{(i-1,i)}\Big)$, 
we specify that $Y_i\cup X_i$ intersects 
$B_{i-1}$ in the set $X_{i-1}:=X_{i-1}'\cup X_{i-1}''$, and
that $\prec_{(i-1,i)}$ contains a linear order $\prec_{i-1}$ on $B_{i-1}$,
which means that we can emulate the formation of the hull
within $G_i$ just by working within $\partial G_i$. 
We fix $\partial Y_i$ in order to determine 
the right choice for $\rho_{i-1}$.

Formally, let ${\cal Y}$ be the set of all triples $\Big(X_{i-1}'',\partial Y_i,\prec_{(i-1,i)}\Big)$, where
\begin{itemize}
\item $X_{i-1}''$ is a subset of $B_{i-1}''$, 
\item $\partial Y_i$ is a subset of $\partial V_i\setminus (B_i\cup B_{i-1})$, 
\item $\Big|X_{i-1}''\cup \partial Y_i\Big|\leq t-|X_i|$, and
\item $\prec_{(i-1,i)}$ is a linear extension of $\prec_i$ to $B_{i-1}\cup B_i$ 
such that $u\prec_{(i-1,i)} v$ 
for every $u\in X_i\cup X_{i-1}''$ 
and every $v\in \Big(B_{i-1}\cup B_i\Big)\setminus \Big(X_i\cup X_{i-1}''\Big)$.
\end{itemize}
Note that ${\cal Y}$ contains $O\Big(2^{t-1}n^t(2t-2)!\Big)$ elements.

We now explain how to choose $\rho_{i-1}$ given an element of ${\cal Y}$.

Let $\Big(X_{i-1}'',\partial Y_i,\prec_{(i-1,i)}\Big)$ be an element of ${\cal Y}$.

Let $v_1\prec_{(i-1,i)} \ldots \prec_{(i-1,i)} v_p$ be the linear order $\prec_{(i-1,i)}$ on $B_{i-1}\cup B_i$.

For every $j$ in $[p]$ with $v_j\in (B_i\cup B_{i-1})\setminus (X_i\cup X_{i-1}'')$,
let $h_j$ be the number of neighbors of $v_j$ in the hull of the set 
$$\Big\{ v_\ell:\ell\in [j-1]\Big\}\cup X_i\cup X_{i-1}''\cup \partial Y_i$$
in 
$\Big(\partial G_i-\Big\{ v_\ell:\ell\in [p]\setminus [j-1]\Big\},\tau\Big).$

If 
$B_i\cup B_{i-1}\cup \partial Y_i$ is not a dynamic monopoly of $(\partial G_i,\tau)$
or
if $h_j<\tau(v_j)-\rho_i(v_j)$
for some $j$ in $[p]$ with $v_j\in B_i\setminus (X_i\cup B_{i-1})$,
then let $f\Big(X_{i-1}'',\partial Y_i,\prec_{(i-1,i)}\Big)=\infty$.
Note that these two cases correspond 
to violations of the conditions (b$'$) and (c$'$)
in the proof of Claim \ref{claim4},
that is, in these cases
there is no set $Y_i$ as in the definition of 
${\rm dyn}_i(X_i,\prec_i,\rho_i)$, 
and, consequently, ${\rm dyn}_i(X_i,\prec_i,\rho_i)=\infty$.

Now, we may assume that 
$B_i\cup B_{i-1}\cup \partial Y_i$ is a dynamic monopoly of $(\partial G_i,\tau)$
and that 
$h_j\geq \tau(v_j)-\rho_i(v_j)$
for every $j$ in $[p]$ with $v_j\in B_i\setminus (X_i\cup B_{i-1})$.
In this case, 
let $f\Big(X_{i-1}'',\partial Y_i,\prec_{(i-1,i)}\Big)$ equal
$$|\partial Y_i|+|X_{i-1}''|+{\rm dyn}_{i-1}\Big(\Big(X_{i-1}'\cup X_{i-1}''\Big),\prec_{i-1},\rho_{i-1}\Big),$$
where 
\begin{itemize}
\item $\prec_{i-1}$ is the restriction of $\prec_{(i-1,i)}$ to $B_{i-1}$,
\item $\rho_{i-1}(v_j)=\rho_i(v_j)+h_j$ for every $j$ in $[p]$ 
with $v_j\in B_{i-1}'\setminus X_{i-1}'$, and
\item $\rho_{i-1}(v_j)=h_j$ for every $j$ in $[p]$ 
with $v_j\in B_{i-1}''\setminus X_{i-1}''$.
\end{itemize}
Note that also in this case $f\Big(X_{i-1}'',\partial Y_i,\prec_{(i-1,i)}\Big)$ can be $\infty$.
Note furthermore that,
for every $v_j\in B_{i-1}'\setminus X_{i-1}'$,
the value of $\rho_{i-1}(v_j)$
has a contributing term $\rho_i(v_j)$ quantifying the help from outside of $V_i$
as well as a contributing term $h_j$ quantifying the help from outside of $V_{i-1}$ but from inside of $V_i$.
For every $v_j\in B_{i-1}''\setminus X_{i-1}''$,
there is no help from outside of $V_i$, that is, the first term disappears.
In view of the above explanation, 
it now follows easily that 
the best choice within ${\cal Y}$
yields ${\rm dyn}_i(X_i,\prec_i,\rho_i)$, that is,
\begin{eqnarray}
{\rm dyn}_i(X_i,\prec_i,\rho_i)&=&
\min\Big\{ f\Big(X_{i-1}'',\partial Y_i,\prec_{(i-1,i)}\Big):
\Big(X_{i-1}'',\partial Y_i,\prec_{(i-1,i)}\Big)\in {\cal Y}\Big\}.
\label{e1}
\end{eqnarray}
In fact, if $Y_i$ is as in the definition of 
${\rm dyn}_i(X_i,\prec_i,\rho_i)$,
and $\prec$ is as in (v) for that set, then 
\begin{eqnarray*}
{\rm dyn}_i(X_i,\prec_i,\rho_i) & = & |Y_i|\\
&=&|\partial Y_i|+|X_{i-1}''|+|Y_{i-1}|\\
& \geq &
|\partial Y_i|+|X_{i-1}''|+
{\rm dyn}_{i-1}\Big(\Big(X_{i-1}'\cup X_{i-1}''\Big),\prec_{i-1},\rho_{i-1}\Big)\\
&=&f\Big(X_{i-1}'',\partial Y_i,\prec_{(i-1,i)}\Big),
\end{eqnarray*}
where 
$\partial Y_i=Y_i\cap (\partial V_i\setminus B_i)$,
$X''_{i-1}=Y\cap B''_{i-1}$,
$Y_{i-1}=Y_i\cap (V_{i-1}\setminus B_{i-1})$, 
$X'_{i-1}=X_i\cap B'_{i-1}$, and
$\prec_{i-1}$ is the restriction of $\prec$ to $B_{i-1}$,
where the inequality follows because the set $Y_{i-1}$ 
satisfies the conditions in the definition of 
${\rm dyn}_{i-1}\Big(\Big(X_{i-1}'\cup X_{i-1}''\Big),\prec_{i-1},\rho_{i-1}\Big)$.

Conversely, if 
$\Big(X_{i-1}'',\partial Y_i,\prec_{(i-1,i)}\Big)$ is in ${\cal Y}$,
and the set $Y_{i-1}$ is as in the definition of 
${\rm dyn}_{i-1}\Big(\Big(X_{i-1}'\cup X_{i-1}''\Big),\prec_{i-1},\rho_{i-1}\Big)$,
then the set 
$Y_i=Y_{i-1}\cup X_{i-1}''\cup \partial Y_i$ 
satisfies the conditions in the definition of 
${\rm dyn}_i(X_i,\prec_i,\rho_i)$,
and, hence,
\begin{eqnarray*}
{\rm dyn}_i(X_i,\prec_i,\rho_i) & \leq & |Y_i|\\
&=&|\partial Y_i|+|X_{i-1}''|+|Y_{i-1}|\\
&=&
|\partial Y_i|+|X_{i-1}''|+
{\rm dyn}_{i-1}\Big(\Big(X_{i-1}'\cup X_{i-1}''\Big),\prec_{i-1},\rho_{i-1}\Big)\\
&=& f\Big(X_{i-1}'',\partial Y_i,\prec_{(i-1,i)}\Big),
\end{eqnarray*}
which shows (\ref{e1}).

Since ${\cal Y}$ has polynomially many elements, 
and each $f\Big(X_{i-1}'',\partial Y_i,\prec_{(i-1,i)}\Big)$ can be determined in polynomial time,
the claim follows.
\end{proof}
Since $k\leq n$,
and there are only polynomially many local cascades 
for each $G_i$,
the Claims \ref{claim3}, \ref{claim4}, and \ref{claim5} complete the proof.
\end{proof}
The algorithm described in the proof of Theorem \ref{theorem2}
can easily be modified in such a way that it also 
determines a minimum dynamic monopoly of $(G,\tau)$ 
within the same time bound.
While many ideas used in this proof extend to chordal graphs,
the number of choices for the linear orders $\prec$ 
seems to be a problem 
for the extension of Theorem \ref{theorem2} to chordal graphs.

We proceed to the proof of our second result.

\begin{proof}[Proof of Theorem \ref{theorem1}]
Since the hull of a set in $(G,\tau)$ can be determined in polynomial time,
the considered problem is in NP.
In order to show hardness, we describe a reduction from the NP-complete problem 
{\sc Vertex Cover} restricted to cubic graphs. Therefore, let $G$ be a cubic graph of order $n$.
Let $G'$ arise from the complete graph $K$ with vertex set $V(G)$
by adding, for every edge $uv$ of $G$, 
a clique $K(uv)$ of order $n$ as well as 
all $2n$ possible edges between $K(uv)$ and $\{ u,v\}$.
Let 
$$\tau:V(G')\to\mathbb{N}_0:
u\mapsto
\begin{cases}
3n+3&\mbox{, if $u\in V(G)$, and}\\
1 &\mbox{, otherwise.}
\end{cases}
$$
In order to complete the proof, 
it suffices to show that the vertex cover number of $G$
equals ${\rm dyn}(G',\tau)$.

First, suppose that $X$ is a vertex cover of $G$.
Let $H$ be the hull of $X$ in $(G',\tau)$.
Since every vertex in $V(G')\setminus V(G)$ has a neighbor in $X$
and threshold value $1$, the set $H$ contains $V(G')\setminus V(G)$.
Therefore, for every vertex $u$ of $G'$ in $V(G)\setminus X$,
the set $H$ contains all three neighbors of $u$ in $V(G)$
as well as all $3n$ neighbors of $u$ in $V(G')\setminus V(G)$,
which implies that $X$ is a dynamic monopoly of $(G',\tau)$.

Next, suppose that $D$ is a dynamic monopoly of $(G',\tau)$.
Since replacing a vertex in $D\setminus V(G)$ by some neighbor in $V(G)$ 
yields a dynamic monopoly, we may assume that $D\subseteq V(G)$.
Suppose, for a contradiction, 
that $u_r,u_s\not\in D$ for some edge $u_ru_s$ in $G$,
where $u_1\prec\ldots\prec u_{n'}$ is a cascade for $D$, and $r<s$.
It follows that $\{ u_j:j\in [r-1]\}$ contains no vertex of $K(u_ru_s)$, 
which implies the contradiction
$|N_{G'}(u_r)\cap \{ u_j:j\in [r-1]\}|\leq 2+2n$.
Hence, $D$ is a vertex cover of $G$,
which completes the proof.
\end{proof}

\end{document}